 \documentclass[letterpaper, 10 pt, journal, twoside]{IEEEtran} 

\usepackage{bm}
\usepackage{enumerate}
\usepackage{commath}
\usepackage{graphicx}
\graphicspath{{Pictures/}}
\usepackage{amsmath}
\usepackage{subcaption}
\usepackage{breqn}
\usepackage{cite}
\usepackage[table]{xcolor}
\usepackage{booktabs}
\usepackage{empheq}
\usepackage{amsfonts}
\usepackage{amssymb}

\usepackage{amsthm}
\usepackage{hyperref}
\usepackage{xcolor}
\usepackage{mathrsfs}
\usepackage{accents}
\usepackage{thmtools}
\usepackage{thm-restate}
\usepackage{float}
\usepackage{xcolor}
\usepackage{comment}
\usepackage[normalem]{ulem}

\theoremstyle{plain}
\newtheorem{lemma}{Lemma}
\newtheorem{definition}{Definition}

\newtheorem{cor}{Corollary}
\newtheorem{proposition}{Proposition}
\newtheorem{remark}{Remark}
\newtheorem{example}{Example}

\newtheorem*{problem*}{Problem}
\newtheorem*{theorem*}{Theorem}
\newtheorem{assumption*}{Assumption}

\declaretheorem[name=Theorem]{thm}

\newcommand{\myvar}[1]{\bm{#1}}
\newcommand{\myvarfrak}[1]{\bm{\mathfrak{#1}}}

\newcommand{\myvardot}[1]{\dot{\myvar{#1}}}

\newcommand{\myset}[1]{\mathscr{#1}}

\title{
High-order  Barrier Functions: Robustness, Safety and Performance-Critical Control
}

\author{Xiao Tan, Wenceslao Shaw Cortez and Dimos V. Dimarogonas 
\thanks{This work was supported by the Swedish Research Council (VR),
the Swedish Foundation for Strategic Research (SSF), the Knut and Alice
Wallenberg Foundation (KAW), and EU CANOPIES project.
The authors are with the School of EECS, Royal Institute of Technology (KTH), 100 44 Stockholm, Sweden (Email: 
        {\tt\small xiaotan, wenscs, dimos@kth.se}).}
}

\begin{document}

\maketitle
\thispagestyle{plain}
\pagestyle{plain}

\begin{abstract}
  In this paper, we propose a notion of high-order (zeroing) barrier functions that generalizes the concept of zeroing barrier functions and guarantees set forward invariance by checking their higher order derivatives. The proposed formulation guarantees asymptotic stability of the forward invariant set, which is highly favorable for robustness with respect to model perturbations. No forward completeness assumption is needed in our setting in contrast to existing high order barrier function methods. For the case of controlled dynamical systems, we relax the requirement of uniform relative degree and propose a singularity-free control scheme that yields a locally Lipschitz control signal and guarantees safety. Furthermore, the proposed formulation accounts for ``performance-critical" control: it guarantees that a subset of the forward invariant set will admit any existing, bounded control law, while still ensuring forward invariance of the set. Finally, a non-trivial case study with rigid-body attitude dynamics and interconnected cell regions as the safe region  is investigated.
 
\end{abstract}

\section{Introduction}

 Optimizing system performance while satisfying safety guarantees is an important goal for controlling dynamical systems. For a general nonlinear system {wherein} an analytical solution is difficult to compute, model predictive control (MPC) and barrier function techniques are two relevant tools to guarantee constraint satisfaction i.e., safety. MPC  \cite{garcia1989model,mayne2000constrained,camacho2013model} is a powerful tool that takes all safety constraints into account at every discrete time instant and solves an optimization problem up to a {finite horizon} with the system performance metric as the objective function. This inevitably brings heavy computational burden for online implementation and the resulting controller {provides} constraint satisfaction and {optimality}. Barrier functions, on the other hand,  provide a system-level certificate that guarantees the forward invariance of a set, {usually referred to} as the ``safety set", that can be designed  in parallel to {a performance-optimizing controller}  \cite{romdlony2014uniting}. This modular formulation gives designers {greater flexibility}.
 
There are {several types of} ``barrier functions'' in the literature. One is related to  \textit{barrier Lyapunov functions} \cite{tee2009barrier} that were introduced  and extensively studied for constrained control problems. Barrier Lyapunov functions are constructed {so that they} tend to infinity when the {system's} state approaches the {boundary of the safety set}. Using backstepping techniques, barrier Lyapunov functions are {extendable} to high-order control systems.  The term ``barrier" is taken from  {optimization} theory \cite{boyd2004convex} {wherein} barrier/penalty terms {are used} to avoid exploration of unwanted regions.  An extension of this methodology, later {coined} \textit{reciprocal barrier functions}\cite{Ames2017}, is presented in \cite{Ames2014}. Reciprocal barrier functions also blow up at the safety boundary and guarantee forward invariance of the safe set if a Lyapunov-like condition holds. Another form of barrier functions, also known as barrier certificates, arise from system verification. {Those barrier certificates are} Lyapunov-like functions {that are used} to verify safety of nonlinear and stochastic systems \cite{prajna2006barrier,prajna2007framework}. In {those methods}, the unsafe region is described by the superlevel set  of a real-valued function and if the derivative of this function is negative definite, then the system is verified to be safe. The controlled version is also discussed in \cite{wieland2007constructive}. {A major} limitation of  reciprocal barrier functions is that {a large control signal is required} when the {system's} state is close to the boundary of the safety set, making it sensitive to noises in the system. {On the other hand, barrier certificates ensure invariance of every level set}, which indicates that the condition imposed is too strong and restrictive.

{Recently}, \cite{Xu2015a,Ames2017} proposed \textit{zeroing barrier functions} (ZBFs) that are well-defined both inside and outside the safe set, and only {ensure} invariance of the safe set. More importantly,  ZBFs provide robustness properties with respect to model perturbations. {Robustness} is {addressed} by {ensuring} asymptotic stability of the forward invariant set and an Input-to-State stability property of the safe set is established. \textit{Zeroing control barrier functions} (ZCBFs), the controlled version of zeroing barrier functions, { originally addressed} relative degree one constraints, and {robustness was further studied in} \cite{jankovic2018robust}. This tool is applicable in a wide range of applications, e.g., in  multi-robot coordination, verification and control \cite{glotfelter2017nonsmooth,Wang2017a,lindemann2018control}.
  
Recently,  \cite{nguyen2016exponential,xu2018constrained,xiao2019control,wences2020correct} have started to investigate conditions on the higher order derivatives of constraint functions to guarantee set invariance. This is motivated by two facts: 1)  by examining the conditions on the high-order  derivative terms, an alternative method  to find barrier functions {is} provided; 2)  many constraints have higher relative degrees with respect to the underlying system, e.g., a position constraint {for} a mechanical system. Thus a systematic framework for higher order barrier functions is highly relevant for real-world applications. Although many important results have been obtained in \cite{nguyen2016exponential,xu2018constrained,xiao2019control,wences2020correct}, we argue that the formulations therein have certain limitations in the sense discussed below and can be considered as special cases of the results presented here. 
 
In this paper, we propose a novel definition of high-order barrier functions  (HOBFs) that generalizes the concept of zeroing barrier functions \cite{Xu2015a,Ames2017}  and the formulations in \cite{nguyen2016exponential,xu2018constrained,xiao2019control,wences2020correct}. In our formulation, extended class $\mathcal{K}$ functions are incorporated instead of linear functions \cite{nguyen2016exponential,xu2018constrained} or class $\mathcal{K}$ functions \cite{xiao2019control}. Apart from this definition generalization,  the contributions of this paper are stated as follows:
 \begin{enumerate}
     \item  In our formulation, the forward completeness assumption in \cite{xu2018constrained,xiao2019control} is no longer required. More importantly,  the forward invariant set is proven to be asymptotically stable for the first time in an HOBF setting and inherits all the robustness properties of ZBFs as in \cite{Xu2015a}. 
     
     \item For the controlled system, we {allow} the relative degree {to vary} in the safe region, which relaxes the uniform relative degree assumption in \cite{xu2018constrained,xiao2019control}. The high-order control barrier function is constructed by introducing a truncating function to the original constraint. The obtained control law is shown to be Lipschitz continuous and the safe set is guaranteed to be forward invariant.  
     
     \item  In many applications,  a pre-designed nominal control law must be implemented without modification in a desired region to ensure satisfaction of the task. This is coined a \textit{performance-critical task}. {Most} ZCBF methods aim {to be} minimally invasive, but do not specify when the nominal control will be implemented \textit{a priori}. {Our formulation allows one to design}  performance-critical regions where the nominal input will be used.
 \end{enumerate}

\textit{Notation}: The Lie derivatives of a function $h(\myvar{x})$ for the system $\myvardot{x} = \myvarfrak{f}(\myvar{x}) + \myvarfrak{g}(\myvar{x}) \myvar{u}$ are denoted by $L_{\mathfrak{f}} h = \tfrac{\partial h}{ \partial \myvar{x}} \myvarfrak{f}(\myvar{x}) $ and $L_{\mathfrak{g}} h = \tfrac{\partial h}{ \partial \myvar{x}} \myvarfrak{g}(\myvar{x})$, respectively.  The notations $\prec, \preceq$ and $\succ, \succeq$ are used to denote element-wise vector inequalities. The interior and boundary of a set $\myset{A}$ are denoted $\text{Int}(\myset{A})$ and $\partial \myset{A}$, respectively. The distance from  a point $\myvar{x}$ to a set $\myset{A}\subset \mathbb{R}^n$ is given by $\| \myvar{x} \|_{\myset{A}} := \inf_{\myvar{w}\in \myset{A}} \| \myvar{x} - \myvar{w} \| $.  The tangent cone  to the set $\myset{A}$ at the point $\myvar{x}$ is defined as $\mathcal{T}_{\myset{A}}(\myvar{x}) : = \{ \myvar{z}: \liminf_{\tau \to 0} \| \myvar{x}+ \tau \myvar{z}  \|_{\myset{A}}/ \tau = 0 \}$.   Denote $ \mathbb{R}^{n}_{+} := \{ \myvar{a} \in \mathbb{R}^n: a_i \ge 0 \} $, where $a_i$ corresponds to the $i$th component of $\myvar{a}$. We note that $ \myvar{x}\in \partial  \mathbb{R}^{n}_{+} $ if $ \myvar{a}^\top \myvar{x} = 0 $ for some nonzero vector $ \myvar{a} \in \mathbb{R}^{n}_{+} $ and $ \myvar{x}\in Int( \mathbb{R}^{n}_{+}) $ if $ \myvar{a}^\top \myvar{x} >0 $ for all nonzero vectors $ \myvar{a} \in \mathbb{R}^{n}_{+} $.

\section{High-order barrier functions}
In this section, we {propose} a novel HOBF definition, which {generalizes} the zeroing barrier functions from \cite{Xu2015a,Ames2017}. The proposed HOBF formulation is more general {than} previous constructions \cite{xu2018constrained,xiao2019control,wences2020correct}, and {is robust} to perturbations.
 
Consider a nonlinear system on $\mathbb{R}^n$,
\begin{equation}\label{eq:auto_dyn}
\dot{\myvar{x}} = \myvarfrak{f}(\myvar{x})
\end{equation}
with $ \myvarfrak{f} $ locally Lipschitz continuous. Denote by $ \myvar{x}(t,\myvar{x}_0) $ the solution of $ \eqref{eq:auto_dyn} $ starting from $\myvar{x}(t_0) = \myvar{x}_0$.  A set $\myset{A} \subset \mathbb{R}^n$ is called \textit{forward invariant}, if for 
any initial condition $ \myvar{x}_0 \in \myset{A}  $, $\myvar{x}(t,\myvar{x}_0)  \in \myset{A}$ for all $t \in I(\myvar{x}_0 )$. Here $I(\myvar{x}_0 )$ denotes the maximal time interval of existence of $\myvar{x}(t,\myvar{x}_0) $.

 Let $h(\myvar{x}): \mathbb{R}^n \to \mathbb{R}$ be a continuously differentiable function. We define the associated sets as $\myset{C}_{h} = \{\myvar{x} \in \mathbb{R}^n: h(\myvar{x}) \geq 0\}$, $\myset{C}_{h,\delta} = \{\myvar{x} \in \mathbb{R}^n: h(\myvar{x}) \geq \delta\}.$

{H}igh-order barrier functions are {dependent on} extended class $\mathcal{K}$ {functions}, which are defined as follows:
\begin{definition}[Extended class $\mathcal{K}$ function  \cite{Ames2017}] \label{def:extended class K}
A continuous function $\alpha:(-b,a) \to (-\infty,\infty)$ for $a,b \in \mathbb{R}_{>0}$ is an \textit{extended class $\mathcal{K}$ function} if it is strictly increasing and $\alpha(0) = 0$.
\end{definition} 
\noindent Note for clarity, the extended class $\mathcal{K}$ functions addressed here will be defined for $a,b = \infty$.

\subsection{High-order barrier functions}
The class of high-order barrier functions considered in this paper is defined as follows. {G}iven a $r^{th}$-order differentiable function $h: \mathbb{R}^n \to \mathbb{R}$, and  sufficiently smooth extended class $\mathcal{K}$ functions $\alpha_1(\cdot),\alpha_2(\cdot), \cdots,\alpha_{r}(\cdot)$, we define a series of functions as
 \begin{equation}\label{eq:a_series_psi}
 \psi_0(\myvar{x}) = h(\myvar{x}), \psi_k(\myvar{x}) = (\tfrac{d}{dt} + \alpha_{k} ) \psi_{k-1}, \  1\le k \le r,
 \end{equation}
{with the corresponding sets:} $ \myset{C}_{\psi_{k-1}}=\{ x:\psi_{k-1}(\myvar{x}) \ge 0 \} $.

\begin{definition}[High-order (zeroing) barrier function] \label{def:hobf}
A $r^{th}$-order differentiable function $ h:\mathbb{R}^n \to \mathbb{R} $ is a \textbf{high-order (zeroing) barrier function} of degree $ r $ for system \eqref{eq:auto_dyn} if there exist differentiable extended class $ \mathcal{K} $ functions $ \alpha_{k}, k = 1,2,\cdots,r $ and an open set $\myset{D}$ with $\myset{C}:=\bigcap_{k = 1}^{r} \myset{C}_{\psi_{k-1}} \subset \myset{D} \subset \mathbb{R}^n $such that
 	\begin{equation}\label{eq:hobf_def}
 	\psi_r(\myvar{x}) \ge 0, \quad \forall \myvar{x} \in \myset{D},
 	\end{equation}
 	with $\psi_k(\myvar{x})$ defined in \eqref{eq:a_series_psi}. 
\end{definition}

\begin{proposition} \label{prop:hobf}
 Consider an autonomous system in \eqref{eq:auto_dyn} and a $r^{th}$ order differentiable function $h:\mathbb{R}^n \to \mathbb{R}$. If $h$ is an HOBF defined on the open set $\myset{D} $ with $ \myset{C}:= \bigcap_{k = 1}^{r} \myset{C}_{\psi_{k-1}} \subset \myset{D}\subset \mathbb{R}^n$, then $\myset{C}$ is forward invariant. 
\end{proposition}

\begin{proof}
For all $\myvar{x} \in \myset{C} \subset\myset{D}$, $\psi_r(\myvar{x}) \ge 0$, we obtain
\begin{equation*} \label{eq:nagumo_vector_dirc}
    \begin{aligned}
    \tfrac{\partial \psi_{k-1}}{\partial \myvar{x}} \myvarfrak{f} (\myvar{x})  & = \tfrac{d}{dt} \psi_{k-1}(\myvar{x}) = - \alpha_{k}(\psi_{k-1}(\myvar{x})) + \psi_{k}(\myvar{x}) \\ &\ge  - \alpha_{k}(\psi_{k-1}(\myvar{x})), 
 \quad   1 \le k \le r, \forall \myvar{x} \in \myset{C}.
    \end{aligned}
\end{equation*} 
We thus have
\begin{equation*}
\tfrac{\partial \psi_{k-1}}{\partial \myvar{x}} \myvarfrak{f} (\myvar{x})  \ge 0, \quad \forall \myvar{x} \in \partial \myset{C}_{\psi_{k-1}} \cap \myset{C} \subset \myset{C},
\end{equation*}
Thus, by definition of the tangent cone,
\begin{equation*} \label{eq:sub_tang_condition}
\myvarfrak{f} (\myvar{x}) \in \mathcal{T}_{\myset{C}_{\psi_{k-1}}}(\myvar{x}), \quad \forall \myvar{x} \in \partial \myset{C}_{\psi_{k-1}} \cap \myset{C} \subset \myset{C}.
\end{equation*}
{Let} $\text{Act}(\myvar{x})$ {denote} the set of active constraints{, i.e., }$\text{Act}(\myvar{x}) = \{k : \psi_{k-1}(\myvar{x}) = 0\}$. Thus, for $\myvar{x} \in \partial  \myset{C}$, the following holds
\begin{equation*}
    \myvarfrak{f} (\myvar{x}) \in \mathcal{T}_{\myset{C}_{\psi_{k-1}}}(\myvar{x}), \quad k\in \text{Act}(\myvar{x})
\end{equation*}
{This implies that} $\myvarfrak{f} (\myvar{x}) \in \mathcal{T}_{\myset{C}}(\myvar{x})$ for all $\myvar{x}\in \myset{C}$. {Since} $\mathfrak{f}$ is  locally Lipschitz,  {the application of} Brezis's Theorem \cite[Theorem~4]{redheffer1972theorems} {ensures that} the set $\myset{C}$ is forward invariant.
\end{proof}

\begin{remark}
 Nagumo's Theorem \cite[Theorem~4.7]{Blanchini2015} has been applied in the barrier function community to {guarantee} forward invariance. However, we need to point out that Nagumo's theorem cannot be applied in the previous proof because, to guarantee  forward invariance, it requires forward completeness of the system \eqref{eq:auto_dyn}, which is not assumed in our case. Instead,  Brezis's theorem {dictates} that with a locally Lipschitz continuous vector field $\mathfrak{f}$ and a closed set $\myset{A} $,  $ \mathfrak{f} (\myvar{x})\in \mathcal{T}_{\myset{A}}$ for all $\myvar{x} \in \myset{A}$ implies that $\myset{A} $ is forward invariant up to the maximal time interval. If we further assume the set $\myset{A}$ {is} compact, {then} the solution remains in $\myset{A}$ for all $t\ge t_0$.

\end{remark}

Definition \ref{def:hobf} and Proposition \ref{prop:hobf} are {generalizations} of similar concepts proposed in \cite{xu2018constrained} and \cite{xiao2019control}. In \cite{xu2018constrained}, each $ \alpha_{k} $ is restricted to the class of linear functions, i.e., $ \alpha_{k}(v)= a_k v, a_k >0, 1\le k \le r $, {whereas our results hold for any extended class-$\mathcal{K}$ function}. {In} \cite{xiao2019control}, {the HOBFs are not well-defined outside of their safe sets due to the restriction to class-$\mathcal{K}$ functions. Here} we let each $ \alpha_{k} $ be an extended class $\mathcal{K}$ function, which is well-defined even if $\psi_{k-1}(\myvar{x}) <0, 1\le k \le r$. {This is important to address robustness as will be shown in the following section.}

\subsection{ Asymptotic stability of the set $\myset{C}$}
Here, we assume {the} system \eqref{eq:auto_dyn} is forward complete to comply with the conditions for asymptotic stability to a set. Before addressing asymptotic stability, we first recall a generalized comparison lemma from \cite{lakshmikantham1989stability}. The {vector} inequalities {used here} are {to be} interpreted component-wise.

\begin{definition}
	A function $ \myvar{p}: \mathbb{R}_{+}  \times \mathbb{R}^n \to \mathbb{R}^n $ is called \textit{quasimonotone nondecreasing} if,  for $ 1 \leq i \leq n$ and all $ 1 \leq j \leq n,  j\neq i $, $x_i = y_i, x_j \le y_j $  implies that
	\begin{equation}
	   p_i(t, \myvar{x}) \le p_i(t, \myvar{y}) 
	\end{equation}
	for the $ i $th component of $ \myvar{p} (t,\cdot)$ and for each $ t $.  

\end{definition}

To {understand} this definition, we present a simple example. {Suppose} $ \myvar{p}(t, \myvar{v}) = A\myvar{v} $. {If} $ \myvar{p} $ {is} quasimonotone nondecreasing, {then} all the off-diagonal elements in $ A $ {must be} nonnegative. {Also} one {can} verify that, if $ \myvar{p} $ is quasimonotone nondecreasing, then  $ \myvar{y} - \myvar{x} \succeq \myvar{0}$, $ \myvar{a}^\top (\myvar{y}-\myvar{x}) = 0 $ for some nonzero vector $ \myvar{a} \in \mathbb{R}^n_{+} $ implies that $ \myvar{a}^\top (\myvar{p}(t,\myvar{y})- \myvar{p}(t,\myvar{x})) \ge 0 $.

\begin{lemma}\cite[Modified from Theorem~1.5.4]{lakshmikantham1989stability} \label{lem:generalized comparison lemma}
	Consider the vectorial differential system 
	\begin{equation}\label{eq:general_sys}
	\tfrac{d}{dt}\myvar{v} = \myvar{p}(t,\myvar{v}), \myvar{v}(t_0) = \myvar{v}_0
	\end{equation}
	where $ \myvar{p}: \mathbb{R}_{+}  \times \mathbb{R}^n \to \mathbb{R}^n $ is quasimonotone nondecreasing and let $ \myvar{r}(t) $ be the maximal solution existing on $ [t_0, \infty) $. Suppose that a continuous function $ \myvar{m}\in C[\mathbb{R}_{+}, \mathbb{R}^n] $ satisfies, for some fixed Dini derivative\footnote{ For a continuous vectorial function $\myvar{m}:\mathbb{R}\to \mathbb{R}^n$, four forms of Dini derivatives of $\myvar{m}$ at $t$  are defined as follows: $D^{+}\myvar{m}(t) = \limsup_{h \to 0^{+}} ( \myvar{m}(t+h) - \myvar{m}(t))/h, D^{-}\myvar{m}(t) = \limsup_{h \to 0^{-}} (\myvar{m}(t+h) - \myvar{m}(t))/h, D_{+}\myvar{m}(t) = \liminf_{h \to 0^{+}} (\myvar{m}(t+h) - \myvar{m}(t))/h,  D_{-}\myvar{m}(t) = \liminf_{h \to 0^{-}} ( \myvar{m}(t+h) - \myvar{m}(t))/h $.   },
	\begin{equation}\label{eq:vectorial_diff_inequality}
	D\myvar{m}(t) \preceq \myvar{p}(t,\myvar{m}(t)), \quad  t\in [t_0, \infty) .
	\end{equation}
	Then, $ \myvar{m}(t_0) \preceq \myvar{v}_0 $ implies 
	\begin{equation}
	\myvar{m}(t) \preceq \myvar{r}(t), \quad  t\in [t_0, \infty) .
	\end{equation}
\end{lemma}
 
  The difference between this Lemma and Theorem 1.5.4 {of} \cite{lakshmikantham1989stability} is that we do not need the domain of $ \myvar{p} $ to be $ \mathbb{R}_{+} \times \mathbb{R}^n_{+} $, nor {do we require} $\myvar{m}(t_0), \myvar{v}_0 $ to be in $  \mathbb{R}^n_{+} $. The proof is almost identical and presented here for completeness.

 \begin{proof}
 	We first introduce an auxiliary system. From \cite[ Theorem~1.5.1]{lakshmikantham1989stability}, we know that for any compact interval $ [t_0,T] $, there exists an $ \myvar{\epsilon}_0 \succ \myvar{0} $ such that for constant vector $ \myvar{\epsilon}, \myvar{0} \prec \myvar{\epsilon}\prec \myvar{\epsilon}_0 $, solutions $ \myvar{v}(t,\myvar{\epsilon}) $ of $ 	\tfrac{d}{dt}\myvar{v} = \myvar{p}(t,\myvar{v}) + \myvar{\epsilon}, \myvar{v}(t_0) = \myvar{v}_0 + \myvar{\epsilon}$
 	exist on $ [t_0, T] $ and $ \underset{\myvar{\epsilon}\to \myvar{0}}{\lim}   \myvar{v}(t,\myvar{\epsilon}) = \myvar{r}(t) $ uniformly on $ [t_0, T] $. From \cite[ Lemma~1.5.1]{lakshmikantham1989stability} and the condition \eqref{eq:vectorial_diff_inequality}, we know that $ D_{-}\myvar{m}(t) \preceq \myvar{p}(t,\myvar{m}(t)), t \geq t_0 $, where $ D_{-}\myvar{m}(t) = \underset{h\to 0^{-}}{\liminf} (\myvar{m}(t+h) - \myvar{m}(t))/h $.
 	
 	It is enough to show that, for arbitrary compact interval $ [t_0,T]$  and sufficiently small $ \myvar{\epsilon}\succ \myvar{0} $, 
 	\begin{equation}\label{eq:intermediate_step}
 	\myvar{m}(t) \prec \myvar{v}(t,\myvar{\epsilon}), \quad t\in [t_0,T]
 	\end{equation}
 	If \eqref{eq:intermediate_step} is not true for some time instant, since  $  \myvar{m}(t_0) \preceq \myvar{v}_0 \prec \myvar{v}_0 + \myvar{\epsilon}= \myvar{v}(t_0,\myvar{\epsilon}) $ and the continuity of $ \myvar{m}(t), \myvar{v}(t,\myvar{\epsilon}) $, there exists a $ t_1 \in [t_0, T]$ such that,  $  \myvar{v}(t,\myvar{\epsilon}) - \myvar{m}(t) \succ \myvar{0}  \quad \text{for all } t \in [t_0,t_1)$
 	and
 	\begin{equation} \label{eq:hit_the_boundaryRn}
 	    \myvar{v}(t_1,\myvar{\epsilon}) - \myvar{m}(t_1) \in \partial \mathbb{R}^n_{+}   .
 	\end{equation}
 	\eqref{eq:hit_the_boundaryRn} means $\myvar{v}(t_1,\myvar{\epsilon}) - \myvar{m}(t_1)$ is at the boundary of $\mathbb{R}^n_{+}$, hence  a nonzero vector $ \myvar{a} \in \mathbb{R}^{n}_{+}$ exists such that $ \myvar{a}^\top (\myvar{v}(t_1,\myvar{\epsilon}) - \myvar{m}(t_1))=0. $ Employing the quasimonotone nondecreasing property of $ \myvar{p} $, it now yields $\myvar{a}^\top (\myvar{p}(t_1,\myvar{v}(t_1,\myvar{\epsilon})) - \myvar{p}(t_1,\myvar{m}(t_1))) \ge 0.$
 	Let $ w(t) = \myvar{a}^\top (\myvar{v}(t,\myvar{\epsilon}) - \myvar{m}(t)), t\in [t_0, t_1) $. Since $ h <0, w(t_1+h) >0 $ (as a result of $ \myvar{v}(t,\myvar{\epsilon}) - \myvar{m}(t) \in Int(\mathbb{R}^{n}_{+}) $ for $ t \in [t_0,t_1) $) and $ w(t_1) = 0 $, we obtain $ 	D_{-} w(t_1) = \underset{h\to 0^{-}}{\liminf} ( w(t_1+h) - w(t_1))/h \le 0$

 	However, from the quasimonotone nondecreasing property, we get $D_{-}w(t_1)  = \myvar{a}^\top ( D_{-} \myvar{v}(t_1,\myvar{\epsilon}) - D_{-}\myvar{m}(t_1)) 
 	     = \myvar{a}^\top ( \tfrac{d}{dt} \myvar{v}(t,\myvar{\epsilon})|_{t = t_1} - D_{-}\myvar{m}(t_1))  = \myvar{a}^\top ( \myvar{p}(t_1, \myvar{v}(t_1,\myvar{\epsilon})) + \myvar{\epsilon}- D_{-}\myvar{m}(t_1)) 	    	>\myvar{a}^\top ( \myvar{p}(t_1, \myvar{v}(t_1,\myvar{\epsilon}))  - D_{-}\myvar{m}(t_1)) \ge 0,$
 	which is contradiction. Hence the proof is complete.
 \end{proof}

Now we proceed to our analysis of the high-order terms in \eqref{eq:a_series_psi}. First we note that for a given set of $\alpha_k$ functions, each $\psi_{k-1}$ is governed by the system dynamics \eqref{eq:auto_dyn}. We can however rearrange these equations as follows: 

\begin{equation} \label{eq: rerrange_psi}
\begin{bmatrix}
\dot{\psi}_0\\
\dot{\psi}_1 \\
... \\
\dot{\psi}_{r-1}
\end{bmatrix} = \begin{bmatrix}
-\alpha_1 (\psi_0) + \psi_1\\
-\alpha_2 (\psi_1) + \psi_2 \\
... \\
-\alpha_r (\psi_{r-1}) + \psi_{r}
\end{bmatrix} 
\end{equation}

Interpreting \eqref{eq: rerrange_psi} as a nonautonomous  system with state variable $ \myvar{\psi} = (\psi_0, \psi_1, \cdots,\psi_{r-1})^\top $ and the time-varying term $ \psi_{r}(\myvar{x}(t)) $, we can re-write \eqref{eq: rerrange_psi} as
 \begin{equation}\label{eq:psi_dyn}
 \tfrac{d}{dt}\myvar{\psi} := \myvar{p}(t,\myvar{\psi}), \quad \myvar{\psi}(t_0) := \myvar{v}_0.
 \end{equation}
 
  A key observation is that the function $ \myvar{p}:\mathbb{R}_{+} \times \mathbb{R}^{r} \to \mathbb{R}^{r}  $ is quasimonotone nondecreasing. This can be {seen} from the fact that, for any $i = 1,2,\cdots, r-1$,  $p_i(t,\myvar{\psi})$, the $i$th component of $\myvar{p}(t,\myvar{\psi})$,  only contains two terms and is increasing with respect to $\psi_i$, the $(i+1)$th component of the vector $\myvar{\psi}$; for $i = r$, $p_r(t,\myvar{\psi}) $ only contains $\psi_{r-1}$, the $r$th component of the vector $\myvar{\psi}$.    A direct application of Lemma \ref{lem:generalized comparison lemma} yields:
  
 \begin{proposition} \label{prop:lower_bounded_system}
 	Let $ \myvar{m} \in C^{1}(\mathbb{R}_{+}, \mathbb{R}^r) $, {and let} $ \myvar{\psi}(t) $ be the solution of \eqref{eq:psi_dyn}. Then  
 	\begin{equation}
 	\tfrac{d}{dt} \myvar{m}(t) \preceq \myvar{p}(t,\myvar{m}) \text{  for  } t\ge t_0, \text{  and  } \myvar{m}(t_0)\preceq \myvar{v}_0
 	\end{equation}
 	implies that
 	\begin{equation}\label{eq:comparison_m_psi}
 	\myvar{m}(t) \preceq \myvar{\psi}(t)  \text{  for  } t\ge t_0.
 	\end{equation}
  
 \end{proposition} 
 \begin{proof}
  Since $\myvar{p}:\mathbb{R}_{+} \times \mathbb{R}^{r} \to \mathbb{R}^{r} $ is quasimonotone nondecreasing and $ \myvar{\psi}(t) $ exists for $ t\in  [0,\infty) $ (as the system \eqref{eq:auto_dyn} is forward complete),  \eqref{eq:comparison_m_psi} follows directly from Lemma \ref{lem:generalized comparison lemma}.
 \end{proof}
  
  We next introduce an auxiliary system 
 \begin{equation}\label{eq:auxiliary_sys}
 \begin{bmatrix}
 \dot{m}_0\\
 \dot{m}_1 \\
 ...\\
 \dot{m}_{r-1}
 \end{bmatrix} = \begin{bmatrix}
 -\alpha_1 (m_0) + m_1\\
 -\alpha_2 (m_1) + m_2 \\
 ...\\
 -\alpha_r (m_{r-1})
 \end{bmatrix}, \quad \myvar{m}(t_0) := \myvar{v}_0.
 \end{equation}
 with the system state $ \myvar{m} = (m_0, m_1, \cdots, m_{r-1})^\top$. Note for $ h $ to be a HOBF, we require $ \psi_r( \myvar{x}) \ge 0$. Thus, the solution of the auxiliary system $\myvar{m}(t)$ satisfies the conditions in Proposition \ref{prop:lower_bounded_system}, and $\myvar{m}(t) \preceq \myvar{\psi}(t)$ for all $t\ge t_0$.
 
  \begin{proposition} \label{prop:AS_hobf}
    If $h$ is an HOBF for the system \eqref{eq:auto_dyn}  and the set $\myset{C}:=\bigcap_{k = 1}^{r} \myset{C}_{\psi_{k-1}}$ is compact, then the set  $\myset{C}$ is asymptotically stable.
  \end{proposition} 
  \begin{proof}
 We first show the following claims.
  
 \noindent \textbf{Claim 1:}  The origin of \eqref{eq:auxiliary_sys} is globally asymptotically stable.

  \begin{proof}
The system \eqref{eq:auxiliary_sys} has a cascaded structure.  {We} define a class of systems $ \Sigma_k, k \in \{1,2,...,r\} ,$
\begin{align*}
    \Sigma_k : \left\{ \begin{array}{l}
        \dot{m}_{k-1} = -\alpha_k(m_{k-1}) + m_{k}, \\
        \dot{m}_{k} = -\alpha_{k+1}(m_{k}) + m_{k+1}, \\
        \cdots,\\
        \dot{m}_{r-1} = -\alpha_{r}(m_{r-1}) 
    \end{array}\right.
\end{align*}
 with the system states $\myvar{m}_k = (m_{k-1}, m_{k}, \cdots, m_{r-1})^\top$ and the initial value drawn from the corresponding components of $\myvar{v}_0$. It is clear that the auxiliary system \eqref{eq:auxiliary_sys} is exactly the system $\Sigma_{1}$.

We prove Claim 1 in an inductive manner. {F}irst we show that {the} system $\Sigma_{r}: \dot{m}_{r-1} = -\alpha_{r}(m_{r-1}) $ is globally asymptotically stable. {T}hen we show that if {the} system $\Sigma_{k}$ is globally asymptotically stable, so is {the} system $\Sigma_{k-1}$.

{For the} radially unbounded, positive definite Lyapunov function $V_{r}(m_{r-1}) = m_{r-1}^2/2$, we obtain $\dot{V}_{r} = - m_{r-1}\alpha_r(m_{r-1})${, which is} is negative definite. Thus, {the} system $\Sigma_{r}$ is globally asymptotically stable.

Assume system  $\Sigma_{k}$ is globally asymptotically stable.  As a result, the system trajectory $\myvar{m}_k(t)$ is bounded.  {For the} Lyapunov candidate $V_{k-1}(m_{k-2}) = m_{k-2}^2/2$. Differentiation of $V_{k-1}$ yields $ \dot{V}_{k-1}(m_{k-2})  = - m_{k-2}\alpha_{k-1}(m_{k-2})  + m_{k-2} m_{k-1}$.
 Since $| m_{k-1}(t)| \leq \| \myvar{m}_k(t) \| $ is bounded, $\lim_{t \to \infty} m_{k-1}(t) = 0$, we obtain that $ m_{k-2}(t)$ is bounded. Thus $\myvar{m}_{k-1}(t) = (m_{k-2}, m_{k-1}, m_{k}, \cdots, m_{r-1})^\top$ is again bounded. From \cite[Corollary 10.3.3]{isidori1999nonlinear}, since  $ \Sigma_{k}$ is globally asymptotically stable, $\dot{m}_{k-2} = -\alpha_{k-1}(m_{k-2}) $ is globally asymptotically stable, and the integral curve of the composite system $ \Sigma_{k-1}$ is forward complete and bounded, we conclude that the system   $ \Sigma_{k-1}$ is also globally asymptotically stable.

By induction, the system $\Sigma_{1}$ is globally asymptotically stable at the origin, which completes the proof.
  \end{proof}
  
  \noindent \textbf{Claim 2:} If the system \eqref{eq: rerrange_psi} is forward complete, then the set $\mathbb{R}_{+}^{r}$ is asymptotically stable with respect to the system \eqref{eq: rerrange_psi}.
  \begin{proof}
  A closed set $\myset{A}$ is asymptotically stable with respect to a forward complete system $\Sigma$ if the set $\myset{A}$ is forward invariant, attractive and uniformly stable\cite{el2007passivity}. Forward invariance of $\mathbb{R}_{+}^{r}$ is obvious by checking the conditions of Brezis's Theorem.  In the following, we show the latter two properties.

    1) Set attraction. For any $\myvar{\psi}(t_0) \notin \mathbb{R}_+^r$, from Proposition \ref{prop:lower_bounded_system}, $\myvar{\psi}(t) \succeq \myvar{m}(t), \forall t\ge t_0$. Following Claim 1, we obtain $  \lim_{t\to \infty} \myvar{\psi}(t) \succeq \lim_{t\to \infty}\myvar{m}(t) = \myvar{0}$,  implying that $\lim_{t\to \infty} \| \myvar{\psi}(t) \|_{\mathbb{R}_{+}^{r}} = 0$. Thus, the set $\mathbb{R}_{+}^{r}$ is attractive. 
    
    2) Set uniform stability. We show this property in an inductive manner.  For $k\in  \{1,2,...,r\} $, denote $\myvar{\psi}_{k} = (\psi_{k-1}, ..., \psi_{r-1})^\top\in \mathbb{R}^{r-k+1}$ and  $\Sigma_k^{\psi}$ the subsystem of \eqref{eq: rerrange_psi} associated with $\myvar{\psi}_{k}$. It is clear that  $\Sigma_1^{\psi}$ is the system in \eqref{eq: rerrange_psi}. 
    
    Consider $k = r$.  From Proposition \ref{prop:lower_bounded_system}, $ \forall t \ge t_0, \psi_{r-1}(t) \geq m_{r-1}(t)$, thus $| \psi_{r-1}(t)|_{\mathbb{R}_+} \leq |m_{r-1}(t)|$. From Claim 1, $ \forall \epsilon>0$, $\exists \delta>0$ such that $ |\psi_{r-1}(t)|_{\mathbb{R}_{+}} \leq |m_{r-1}(t)|\leq \epsilon, \forall  |\psi_{r-1}(t_0)|_{\mathbb{R}_{+}} =  |m_{r-1}(t_0)|\leq \delta$, i.e., $\mathbb{R}_{+}$ is uniformly stable with respect to $\Sigma_r^{\psi}$.
    
    For $k \in \{2,...,r\}$, assume that $\mathbb{R}_{+}^{r-k+1}$ is uniformly stable  with respect to $\Sigma_{k}^{\psi}$ with the state $\myvar{\psi}_{k}$. For any given $\epsilon>0$, let $\epsilon^\prime >0$ such that $\alpha_{k-1}^{-1}(\epsilon^\prime) + \epsilon^\prime < \epsilon$. By assumption, there exists a $\delta^{\prime} >0$ such that  $  \| \myvar{\psi}_{k} (t)\|_{\mathbb{R}_{+}^{r-k+1}} \leq \epsilon^\prime, \forall \| \myvar{\psi}_{k} (t_0)\|_{\mathbb{R}_{+}^{r-k+1}} \leq \delta^{\prime}, \forall t\ge t_0$. Choose $\delta = \min(\delta^\prime,\alpha_{k-1}^{-1}(\epsilon^\prime))$.  For all $\| \myvar{\psi}_{k-1} (t_0) \|_{\mathbb{R}_{+}^{r-k+2}} \leq  \delta$,  we have $\psi_{k-2}(t_0) \ge - \alpha_{k-1}^{-1}(\epsilon^\prime)$ and $\| \myvar{\psi}_{k} (t_0) \|_{\mathbb{R}_{+}^{r-k+1}} \leq   \delta^{\prime}$, which implies $\| \myvar{\psi}_{k} (t) \|_{\mathbb{R}_{+}^{r-k+1}} \leq   \epsilon^{\prime}$ and $\psi_{k-1}(t) \ge  - \epsilon^\prime$ for $t\ge t_0$. Recall $\dot{\psi}_{k-2} = -\alpha_{k-1}(\psi_{k-2}) + \psi_{k-1}(t)$.  Since  $\dot{\psi}_{k-2}(t) \geq 0$ whenever $\psi_{k-2}(t) = - \alpha_{k-1}^{-1}(\epsilon')$ and $\psi_{k-2}(t_0) \geq - \alpha_{k-1}^{-1}(\epsilon')$, we obtain $\psi_{k-2}(t) \geq - \alpha_{k-1}^{-1}(\varepsilon'), \forall t \geq t_0$. Furthermore,  $ \| \myvar{\psi}_{k-1}(t) \|_{\mathbb{R}_{+}^{r-k+2}} = \| (\psi_{k-2}(t), \myvar{\psi}_{k}^\top(t))^\top \|_{\mathbb{R}_{+}^{r-k+2}} \leq \alpha_{k-1}^{-1}(\epsilon^\prime) + \epsilon^\prime \leq \epsilon, \forall t\ge t_0$. Thus  $\mathbb{R}_{+}^{r-k+2}$ is uniformly stable with respect to $\Sigma_{k-1}^{\psi}$. 
    
    Since $\mathbb{R}_+$ is uniformly stable with respect to $\Sigma_r^\psi$, then applying the previous analysis for $k = r $ ensures that $\mathbb{R}_+^{2}$ is uniformly stable with respect to $\Sigma_{r-1}^\psi$. By repeating this analysis for $k \in \{2,...,r-1\}$,  $\mathbb{R}_{+}^{r}$ is thus uniformly stable with respect to $\Sigma_{1}^{\psi}$, i.e., the system in  \eqref{eq: rerrange_psi}.
  \end{proof}
  
Now we proceed to show the asymptotic stability of the forward invariant set $\myset{C}$. Since the system \eqref{eq:auto_dyn} is forward complete, and  $\myvar{\psi}(\myvar{x})$ is well-defined in $\mathbb{R}^n$, then Claim 2 is applicable. Define $\myset{N}_{a}:=\{ \myvar{x}: \| \myvar{x} \|_{\myset{C}} \leq a\}$ for any $a >0$.

     1) Set uniform stability. Given any $\epsilon>0$ such that $ \myset{N}_{\epsilon} \subset \myset{D}$, we can take $\epsilon^\prime, \delta^\prime >0$ such that $ \epsilon^\prime \in (0,\min_{\|\myvar{x} \|_{\myset{C}} = \epsilon}  \| \myvar{\psi}(\myvar{x}) \|_{\mathbb{R}_{+}^{r}})$, and $ \| \myvar{\psi}(t) \|_{\mathbb{R}_{+}^{r}}\leq \epsilon^\prime, \forall \| \myvar{\psi}(t_0) \|_{\mathbb{R}_{+}^{r}}\leq \delta^\prime, \forall t \ge t_0$.  Here $\myset{N}_{\epsilon}$, the minimum exist since $\myset{C}$ is compact. The $\epsilon^\prime, \delta^\prime$ pair always exists following Claim 2. Based on the continuity and positive semi-definiteness of the function $\myvar{x} \mapsto \| \myvar{\psi}(\myvar{x}) \|_{\mathbb{R}_{+}^{r}}  $, there exists a $ \delta>0$ such that $ \| \myvar{\psi}(\myvar{x}) \|_{\mathbb{R}_{+}^{r}} \leq \delta^\prime, \forall \myvar{x} \in \myset{N}_{\delta}$. Thus, $\forall \myvar{x}_0 \in \myset{N}_{\delta}$, $\| \myvar{\psi}(\myvar{x}_0) \|_{\mathbb{R}_+^r} \leq \delta^\prime$, $\| \myvar{\psi}(\myvar{x}(t,\myvar{x}_0)) \|_{\mathbb{R}_+^r} \leq \epsilon^\prime$, which further implies that $ \myvar{x}(t, \myvar{x}_0) \in \myset{N}_{\epsilon}$ for $t\ge t_0$. Thus $\myset{C}$ is uniformly stable.
     
    2) Set attraction. Choose $\epsilon, \delta >0$ such that  $\myvar{x}(t,\myvar{x}_0) \in \myset{N}_{\epsilon} \subset \myset{D}, \forall \myvar{x}_0\in \myset{N}_{ \delta}, \forall t\ge t_0$ from previous analysis. For any given $\epsilon^\prime \in (0,\epsilon) $, choose $a \in (0,\min_{\|\myvar{x} \|_{\myset{C}} = \epsilon^\prime}  \| \myvar{\psi}(\myvar{x}) \|_{\mathbb{R}_{+}^{r}})$.  Following Claim 2, $\lim_{t\to \infty} \| \myvar{\psi}(\myvar{x}(t,\myvar{x}_0)) \|_{\mathbb{R}_{+}^{r}} = 0, \forall \myvar{x}_0 \in \myset{N}_{\delta}$.  There exists $T>0$ such that $\forall t >T, \forall \myvar{x}_0 \in \myset{N}_{\delta}$, $\myvar{x}(t,\myvar{x}_0) \in \Omega_{a} := \{ \myvar{x}\in \myset{D} : \| \myvar{\psi}(\myvar{x})  \|_{\mathbb{R}_{+}^{r}} \leq a  \} \subset \myset{N}_{\epsilon^\prime} $. With a diminishing $\epsilon^\prime$, we show that $\lim_{t\to \infty} \| \myvar{x}(t,\myvar{x}_0)) \|_{\myset{C}} = 0, \forall \myvar{x}_0 \in \myset{N}_{\delta}$. Thus $\myset{C}$ is attractive.  

  Thus, the set $\myset{C}$ is asymptotically stable.
    \end{proof}

\begin{remark}
Inspired by \cite{el2007passivity}, the condition on $\myset{C}$ being compact can be relaxed, but with the extra assumption that there exist class $\mathcal{K}$ functions $\alpha, \beta$ such that
\begin{equation}
    \alpha(\| \myvar{x} \|_{\myset{C}}) \leq \| \myvar{\psi}(\myvar{x}) \|_{\mathbb{R}_{+}^{r}} \leq  \beta(\| \myvar{x} \|_{\myset{C}})
\end{equation}
for all $\myvar{x} \in \myset{D}$. The proof is omitted due to space limitations. 
\end{remark}

	 Proposition \ref{prop:AS_hobf} generalizes  the asymptotic stability {results} of the set $\myset{C}$ {for relative-degree one ZBFs} \cite[Proposition~4]{Xu2015a}  to HOBFs. This property is beneficial in practice because it indicates several different robustness properties.  As discussed in  \cite{Xu2015a}, for the perturbed system $\dot{\myvar{x}} = \myvarfrak{f}(\myvar{x}) + \myvarfrak{g}(\myvar{x})$, if  $\myvarfrak{g}(\myvar{x})$ is a vanishing perturbation, i.e., $\myvarfrak{g}(\myvar{x})$ is continuous and  satisfies $\| \myvarfrak{g}(\myvar{x}) \| \le \sigma(\| \myvar{x}\|_{\myset{C}})$ for $\myvar{x} \in \myset{D}\setminus \myset{C}$ and some class $\mathcal{K}$ function $\sigma(\cdot)$, then the set $\myset{C}$ is still asymptotically stable. If $\myvarfrak{g}(\myvar{x})$ is not vanishing but sufficiently small, i.e., there exists a positive constant $k$ such that  $\| \myvarfrak{g}(\myvar{x}) \|_{\infty} \le k $, then  a new asymptotically stable set containing $\myset{C}$ as well as asymptotic convergence to this new set can be established. Interested readers can refer to \cite{Xu2015a} and the references therein for more details.

\section{High-order control barrier functions} \label{sec:hocbf}
Consider the nonlinear control affine system
\begin{equation} \label{eq:nonlinear_dyn}
    \myvardot{x} = \myvarfrak{f}(\myvar{x}) + \myvarfrak{g}(\myvar{x}) \myvar{u},
\end{equation}
{with} the state $\myvar{x} \in \mathbb{R}^n$, and the control input $ \myvar{u} \in U \subset \mathbb{R}^m$.   We will consider the simplified case where $\myvarfrak{f}$ and $\myvarfrak{g}$ are locally Lipschitz functions in $\myvar{x}$.

\begin{definition}[Least relative degree] \label{def:least_relative_degree} 
Given an arbitrary set $\myset{D}\subset \mathbb{R}^n$. A $r^{th}$-order differentiable function $h: \mathbb{R}^n \to \mathbb{R}$ has \textit{least relative degree} $r$ in $\myset{D}$  for system \eqref{eq:nonlinear_dyn} if $        L_{\mathfrak{g}} L_{\mathfrak{f}}^k h(\myvar{x}) = \myvar{0},  \forall \myvar{x} \in \myset{D} $
for  $k=1,2,\cdots,r-2$.

\end{definition}
The least relative degree condition is much weaker compared to the uniform relative degree condition {\cite{xiao2019control}}, since the  latter further requires  $L_{\mathfrak{g}} L_{\mathfrak{f}}^{r-1} h(\myvar{x}) \neq \myvar{0},  \forall \myvar{x} \in \myset{D}$.

Formally, a high-order control barrier function is defined as follows:

\begin{definition}[High-order (zeroing) control barrier function (HOCBF)] \label{def:high-order control barrier function} 
Consider control system \eqref{eq:nonlinear_dyn}, and  a $r^{th}$-order differentiable function $h: \mathbb{R}^n \to \mathbb{R}$. The function $h$ is called a \textbf{high-order (zeroing) control barrier function} (of order $r$), if there exist differentiable extended class $\mathcal{K}$ functions {$\alpha_{k}$, $k\in \{1,...,r\}$,} and an open set $\myset{D}$ with $ \myset{C}:=\bigcap_{k = 1}^{r} \myset{C}_{\psi_{k-1}}\subset \myset{D} \subset \mathbb{R}^n$, where $\psi_{k}$ is given in \eqref{eq:a_series_psi}, such that
\begin{enumerate}
    \item $h$ is of least relative order $r$ in $\myset{D}$;
    \item    for all $ \myvar{x} \in \myset{D} $,
\begin{multline} \label{eq:condition_HOCBF}
\underset{\myvar{u} \in \myset{U}}{\sup} \psi_r (\myvar{x}) = \underset{\myvar{u} \in \myset{U}}{\sup} [L_{\mathfrak{f}}\psi_{r-1} (\myvar{x}) 
+ L_{\mathfrak{g}} \psi_{r-1}(\myvar{x}) \myvar{u} \\ + \alpha_r(\psi_{r-1}(\myvar{x}))] \geq 0. 
\end{multline} 
\end{enumerate}
\end{definition}

When letting $r = 1$, an HOCBF yields {the} zeroing control barrier function {of} \cite{Xu2015a}. This definition is also more general to its counterparts in \cite{xu2018constrained} and \cite{xiao2019control} {since} : 1) $\alpha_k$ in \cite{xu2018constrained} is restricted  to {the set of} linear functions, while $\alpha_k$ in  \cite{xiao2019control} {is restricted to the set of} class $\mathcal{K}$ functions. {We note that class-$\mathcal{K}$ functions} are not well-defined for  $\myvar{x}\in \myset{D}\setminus \myset{C}$, {and} thus the robustness results presented here cannot be {applied to the barriers of \cite{xiao2019control}}; 2) {the} uniform relative degree $r$ {condition} is not needed {here}, {and thus our formulation is less restrictive than that of \cite{xiao2019control}}; 3) while \cite{xu2018constrained} and \cite{xiao2019control} both assume the closed-loop system \eqref{eq:nonlinear_dyn} to be forward complete {to ensure forward invariance}, this is not required here. Hereafter we denote $\alpha = \alpha_r$ for notational brevity.

Similar to Proposition \ref{prop:hobf}, the following result guarantees the forward invariance of $\myset{C}$. Given an HOCBF $h$, for all $\myvar{x} \in \myset{D}$, {we} define the set
\begin{multline} \label{eq:k_HOCBF}
    K_{HOCBF}(\myvar{x}) = \{ \myvar{u}\in U: L_{\mathfrak{f}} \psi_{r-1} (\myvar{x}) \\
    + L_{\mathfrak{g}} \psi_{r-1}(\myvar{x}) \myvar{u} + \alpha(\psi_{r-1}(\myvar{x})) \geq 0 \}.
\end{multline}

\begin{thm}\label{thm:forw_inv_HOCBF}
Consider an HOCBF $h$, $\psi_{k-1}, 1\leq k\leq r$ defined in \eqref{eq:a_series_psi}. Then any locally Lipschitz continuous controller $\myvar{u}: \mathbb{R}^n \to \mathbb{R}^m$ such that $\myvar{u}(\myvar{x}) \in K_{HOCBF}$ will render the set $\myset{C}:=\bigcap_{k = 1}^{r} \myset{C}_{\psi_{k-1}}$ forward invariant for {the} system \eqref{eq:nonlinear_dyn}.
\end{thm}
\begin{proof}
The proof follows directly from Proposition \ref{prop:hobf}.
\end{proof}

\begin{remark}
If there exists an HOCBF $h$ and a locally Lipschitz continuous controller $\myvar{u}: \mathbb{R}^n \to \mathbb{R}^m$ such that $\myset{C}$ is compact, $\myvar{u}(\myvar{x}) \in K_{HOCBF}$, and \eqref{eq:nonlinear_dyn} forward complete, then the set $\myset{C}$ is asymptotically stable. This follows directly from the proof of Proposition \ref{prop:AS_hobf}. This property is useful in practice because, for example,  when  the system starts outside of the safe set $\myset{D}\setminus \myset{C}$, we know the system state will asymptotically reach the set $\myset{C} $.
\end{remark}

\begin{remark}
Consider the perturbed system
\begin{equation}
    \dot{\myvar{x}} = \mathfrak{f}(\myvar{x}) + \mathfrak{g}(\myvar{x}) \myvar{u} + \mathfrak{p}(x)\myvar{\omega},
\end{equation}
where  $\myvar{\omega} \in \mathbb{R}^{v}$ is  an external disturbance, while $  \mathfrak{p}(x)\myvar{\omega}$ represents a structured disturbance/uncertainty that is nether vanishing nor sufficiently small. If $        L_{\mathfrak{p}} L_{\mathfrak{f}}^k h(\myvar{x}) = \myvar{0},  \forall \myvar{x} \in \myset{D} $ for  $k=1,2,\cdots,r-2$ (i.e., $h$ has the same least relative degree with respect to $\myvar{\omega}$ as with respect to $\myvar{x}$), then we could robustify the HOCBF condition using a similar technique to \cite{jankovic2018robust} by requiring $L_{\mathfrak{f}}\psi_{r-1} (\myvar{x}) 
+ L_{\mathfrak{g}} \psi_{r-1}(\myvar{x}) \myvar{u} + \| L_{\mathfrak{p}} \psi_{r-1}(\myvar{x}) \| \bar{\omega}  + \alpha_r(\psi_{r-1}(\myvar{x}))\geq 0$, where $ \bar{\omega}$  is the known upper bound of $\myvar{\omega}(t) $. If this condition holds, then the set $\myset{C}$ is again rendered forward invariant for the perturbed system. The proof also  follows  directly  from  Proposition  \ref{prop:hobf}.
		
\end{remark}

{Motivated by existing methods \cite{Ames2019control}, we define} a point-wise minimum-invasive controller. Suppose that a nominal control input $ \myvar{u}_{\text{nom}}: \myset{D} \to \mathbb{R}^m$ Lipschitz continuous in $ \myvar{x} $, has been designed, and we need to modify the control input online to account for the safety constraints. {T}he modified controller is given by the quadratic program below:
\begin{equation} \label{eq:controllerQP}
\begin{aligned}
& \myvar{u}(\myvar{x}) = \arg   \min_{\myvar{u} \in U }\|\myvar{u} - \myvar{u}_{\text{nom}}\|_2^2 \\
   & \hspace{-10pt} \text{s.t.} \quad L_{\mathfrak{g}} \psi_{r-1}( \myvar{x}) \myvar{u}  + L_{\mathfrak{f}} \psi_{r-1}(\myvar{x}) + \alpha(\psi_{r-1}(\myvar{x}))\ge 0.
\end{aligned}
\end{equation}
  This formulation is known as ``safety-critical" in that constraint satisfaction is prioritized over the nominal control law.

\section{Singularity-free, Performance-critical HOCBFs}

In the previous section, the existence of an HOCBF ensures safety of the overall system. However the construction of the HOCBF is not straightforward in general. Following a similar analysis to Section 3.1 of \cite{Xu2015a}, for any $r^{th}$-order differentiable function $h: \mathbb{R}^n \to \mathbb{R}$, if $U = \mathbb{R}^m$ and $ L_{\mathfrak{g}}L_{\mathfrak{f}}^{r-1} h( \myvar{x}) \neq \myvar{0}, \forall \myvar{x} \in \myset{D}$ (i.e., $h$ is of uniform relative degree $r$ in $\myset{D}$), then \eqref{eq:controllerQP} is feasible for all $ \myvar{x} \in \myset{D}$ and $h$ is an HOCBF. Moreover, the resulting controller is locally Lipschitz continuous in $\myset{D}$. In the following section, we will study the case when $ U = \mathbb{R}^m, L_{\mathfrak{g}}L_{\mathfrak{f}}^{r-1} h( \myvar{x}) = \myvar{0}$ for some $ \myvar{x} \in \myset{D}$ (i.e., $h$ is of \textit{least} relative degree $r$ in $\myset{D}$).

\subsection{Singularity-free HOCBF design }  

One notable difference between Definition \eqref{def:high-order control barrier function} and the existing constructions  \cite{xu2018constrained,xiao2019control,wences2020correct} is that an HOCBF candidate does not need to have uniform relative degree $r$. The motivation {for this} comes from the fact that even the double integrator dynamics with circular region constraints will violate this assumption, as {shown} in the following example.

\begin{example} \label{exp:circular_constraint}
Consider the double integrator dynamics $\begin{psmallmatrix}
\myvardot{p} \\
\myvardot{v}
\end{psmallmatrix} =  \begin{psmallmatrix}
\myvar{v} \\
\myvar{0}
\end{psmallmatrix} +  \begin{psmallmatrix}
\myvar{0} \\
I_2
\end{psmallmatrix} \myvar{u}$
with $\myvar{p}, \myvar{v}, \myvar{u}\in \mathbb{R}^2$, $\myvar{x} = (\myvar{p}, \myvar{v})$. Let $b( \myvar{p},\myvar{v}) := d^2 - \lVert \myvar{p} \rVert^2$ defining a circular region  in $\mathbb{R}^2$ with radius $d $.  $\myset{C}_{b} = \{ (\myvar{p},\myvar{v}): b(\myvar{x}) \ge 0  \}$. With straightforward calculation, we obtain $L_{\mathfrak{g}} b = \myvar{0}, L_{\mathfrak{g}}L_{\mathfrak{f}}b = -2\myvar{p}^\top.$
Thus, $ L_{\mathfrak{g}}L_{\mathfrak{f}}b(\myvar{x}) = \myvar{0}$ for  $\myvar{x}\in \myset{E} = \{ (\myvar{p},\myvar{v}): \myvar{p} =  \big(\begin{smallmatrix}
  0\\
  0
\end{smallmatrix}\big) \} \subset \myset{C}_b$, which does not satisfy the conditions from \cite{xiao2019control,wences2020correct}. We will show how the proposed HOCBF considered here addresses the singularity issue for application {to} more general systems/constraints.
\end{example}

We now present a method to address the possible infeasibility of  the quadratic program \eqref{eq:controllerQP} due to the existence of singular points. {In} the following, we show that as long as the singular points are strictly bounded away from the boundary, a novel control barrier function  can be constructed such that the constraints in \eqref{eq:controllerQP} are always feasible.

\begin{proposition} \label{prop:h_hocbf}
 Consider a smooth function $b:\mathbb{R}^n \to \mathbb{R}$ with the associated set $\myset{C}_b$ and an open set $\myset{D}$ with $\myset{C}_b \subset \myset{D}$. Let $b$ have least relative degree $r$ in $\myset{D}$ and define the set $\myset{E} := \{ \myvar{x} \in \myset{D}:  L_{\mathfrak{g}} L_{\mathfrak{f}}^{r-1} b(\myvar{x}) = \myvar{0} \}$. Assume that there exists a scalar $\xi > 0$ such that  
\begin{equation} \label{eq:bounded_away_from_safety_region}
    \myset{E} \subseteq \myset{C}_{b,\xi}.
\end{equation}

Define $h:\mathbb{R}^n \to \mathbb{R}$ as
\begin{equation} \label{eq:h_definition}
    h(\myvar{x}) = \chi\left(\tfrac{b(\myvar{x})}{\xi} \right),
\end{equation}
with $\chi: \mathbb{R} \to \mathbb{R}$ a $r^{th}$-order differentiable function satisfying 
\begin{equation} \label{eq:chi_property}
\left\{
\begin{array}{cc}
\chi(0) = 0, \\
\chi(\tau) = 1, & \text{ for } \tau \ge 1,  \\
 \tfrac{d\chi}{d\tau} (\tau)>0, &  \text{ for } \tau < 1.    
\end{array}\right.
\end{equation}

If $ U = \mathbb{R}^m$, then the function $h$ is an HOCBF.
\end{proposition}

\begin{proof}
  It is trivial to verify that $ \myset{C}_{h} = \myset{C}_{b}:=\{\myvar{x} \in \mathbb{R}^n: b(\myvar{x})\ge 0\} $, and $r$ is also the least relative degree of function $h$. We need to prove that there always exist a $\myvar{u} \in \mathbb{R}^m$ and sufficiently smooth extended class $\mathcal{K}$ functions $\alpha_{k}$s such that 
\begin{equation} \label{eq:cbf_constraint}
    L_{\mathfrak{g}} \psi_{r-1} \myvar{u} + L_{\mathfrak{f}} \psi_{r-1} + \alpha(\psi_{r-1}) \ge 0
\end{equation}
 holds for all $\myvar{x}\in \myset{D}$ with $\psi_{k-1}, k =  1,2, \cdots, r$ defined in \eqref{eq:a_series_psi}. Denote    $ \myset{C}:=\bigcap_{k = 1}^{r} \myset{C}_{\psi_{k-1}}$ as in Definition \ref{def:high-order control barrier function}.
 
 We first examine the properties of $L_{\mathfrak{g}} \psi_{r-1}$ and $L_{\mathfrak{f}} \psi_{r-1}$. With $h(\myvar{x})$ defined in \eqref{eq:h_definition}, we obtain $L_{\mathfrak{f}}h  = \tfrac{d\chi}{d\tau}(b(\myvar{x})/\xi) \tfrac{\partial b/\xi}{\partial \myvar{x}}\cdot \mathfrak{f} = \tfrac{1}{\xi}\tfrac{d\chi}{d\tau} L_{\mathfrak{f}}b,
         L_{\mathfrak{g}} h =  \tfrac{1}{\xi}\tfrac{d\chi}{d\tau} L_{\mathfrak{g}}b.$
If $r > 1$, then $  L_{\mathfrak{g}} h = \myvar{0}$. Note that $\psi_1 = L_{\mathfrak{f}}h + \alpha_1(h(\myvar{x})) $, it derives
\begin{equation}
    \begin{aligned}
        L_{\mathfrak{f}}\psi_1 &  = \tfrac{1}{\xi} \big(  \tfrac{1}{\xi} \tfrac{d^2 \chi}{d \tau^2} L_{\mathfrak{f}}b L_{\mathfrak{f}}b + \tfrac{d\chi}{d\tau}L^2_{\mathfrak{f}}b \big) + \tfrac{d\alpha_1}{dh} L_{\mathfrak{f}}h \\
        L_{\mathfrak{g}} \psi_1 & =  \tfrac{1}{\xi} \big( \tfrac{1}{\xi} \tfrac{d^2 \chi}{d \tau^2}  L_{\mathfrak{f}}b   \tfrac{\partial b}{ \partial \myvar{x}} \cdot \mathfrak{g} + \tfrac{d \chi}{d \tau} L_{\mathfrak{g}} L_{\mathfrak{f}}b \big) + \tfrac{d\alpha_1}{d h} L_{\mathfrak{g}} h \\
        & = \tfrac{1}{\xi} \tfrac{d\chi}{d\tau}L_{\mathfrak{g}}L_{\mathfrak{f}}b
    \end{aligned}
\end{equation}
If $r>2$, then $L_{\mathfrak{g}} \psi_1 = \myvar{0}$ and we can iterate these calculations until $\psi_{r-1}$ that gives us $       L_{\mathfrak{g}}\psi_{r-1}  = \tfrac{1}{\xi} \tfrac{d\chi}{d\tau}L_{\mathfrak{g}}L_{\mathfrak{f}}^{r-1} b.$
Thus, in view of the properties of $\chi$ given in \eqref{eq:chi_property},  we know
\begin{enumerate}
    \item $ L_{\mathfrak{g}}\psi_{r-1} (\myvar{x}) = \myvar{0}$ if and only if  $\myvar{x} \in \myset{D}\cap \myset{C}_{b,\xi}$;
    \item $ L_{\mathfrak{f}}\psi_{r-1}(\myvar{x})  = 0$ if $ \myvar{x} \in \myset{D}\cap \myset{C}_{b,\xi}$.
\end{enumerate}
 
The condition in \eqref{eq:cbf_constraint} is examined in  two cases. For $\myvar{x}\in \myset{D}\cap \myset{C}_{b,\xi}$, we derive that $L_{\mathfrak{g}}\psi_{r-1}(\myvar{x}) = \myvar{0}$,  $L_{\mathfrak{f}}\psi_{r-1}(\myvar{x}) = 0, \alpha(\psi_{r-1}) = \alpha( \alpha_{r-1}(\psi_{r-2})) = \cdots = \alpha_{r} \circ \alpha_{r-1} \circ \cdots \alpha_{1}(h(\myvar{x})) = \alpha_{r} \circ \alpha_{r-1} \circ \cdots \alpha_{1}(1) > 0 $, thus the condition in \eqref{eq:cbf_constraint} is trivially satisfied. For $ \myvar{x} \in \myset{D}\setminus (\myset{D} \cap \myset{C}_{b,\xi}) $, as $L_{\mathfrak{g}} \psi_{r-1} \neq \myvar{0}$ and the condition in \eqref{eq:cbf_constraint} imposes a linear constraint on $\myvar{u}$. {Thus}, there always exists a $\myvar{u} \in \mathbb{R}^m$ that satisfies \eqref{eq:cbf_constraint} and $h$ is an HOCBF via Definition \ref{def:high-order control barrier function}. 
\end{proof}

Here we note that the Assumption in \eqref{eq:bounded_away_from_safety_region} is intuitive and easy-to-check as it requires all the singularity points {to be} inside $\myset{C}_{b,\xi}$ for some positive number $\xi$. In the double integrator example, this assumption is clearly fulfilled as $\myset{E}  \subseteq \myset{C}_{b,\xi}$ for any $ 0<\xi<d $. We further show that the resulting controller is locally Lipschitz continuous. 

\begin{proposition}\label{prop:lipschitz_continuity_solution}
  Assume  the conditions in Proposition \ref{prop:h_hocbf}  hold and the nominal controller $\myvar{u}_{nom} : \myset{D} \to \mathbb{R}^m$ is  bounded and locally Lipschitz continuous in $\myset{D}$. With $h$ given in \eqref{eq:h_definition} and $\psi_{k}$ given in \eqref{eq:a_series_psi}, assume furthermore that $L_{\mathfrak{g}} \psi_{r-1}$ and $L_{\mathfrak{f}} \psi_{r-1}$ are locally Lipschitz continuous. Then, 
   \begin{enumerate}
        \item  the solution to the quadratic program  \eqref{eq:controllerQP} is locally Lipschitz continuous in $\myset{D}$;
     \item  the controller \eqref{eq:controllerQP} renders the set $\myset{C}:= \bigcap_{k =1}^{r}  \myset{C}_{\psi_{k-1}} $ forward invariant for system in \eqref{eq:nonlinear_dyn}.
 \end{enumerate}
\end{proposition}

\begin{proof}
	The feasibility of the linear inequality constraint on $ \myvar{u} $ is guaranteed in Proposition \ref{prop:h_hocbf} for every $ \myvar{x} \in \myset{D}  $. The solution to the quadratic program  \eqref{eq:controllerQP} has a closed-form solution, given by the KKT condition \cite{boyd2004convex}, as
	\begin{equation} \label{eq:controllerSol}
	\begin{aligned}
	\myvar{u}(\myvar{x}) = \myvar{u}_{\text{nom}}(\myvar{x}) +   \mu L_{\mathfrak{g}}^\top \psi_{r-1}(\myvar{x})
	\end{aligned}
	\end{equation}
	with 
	\begin{multline*}
	   	\mu  = \left\{ \begin{array}{l}
	0,   \quad  \text{   if }   L_{\mathfrak{g}}\psi_{r-1} \myvar{u}_{\text{nom}} + \alpha(\psi_{r-1}) + L_{\mathfrak{f}}\psi_{r-1} \ge 0 , \\
	   \dfrac{- L_{\mathfrak{g}}\psi_{r-1} \myvar{u}_{\text{nom}} - \alpha(\psi_{r-1}) - L_{\mathfrak{f}}\psi_{r-1}}{\| L_{\mathfrak{g}}\psi_{r-1} \|^2},  \text{ otherwise. }
	\end{array}\right. 
	\end{multline*}
The derivation is straightforward considering whether the linear constraint on $\myvar{u}$ in \eqref{eq:controllerQP} is active or not and thus omitted here. Recall that $ L_{\mathfrak{g}}\psi_{r-1} = \myvar{0} $ if and only if $ \myvar{x}\in  \myset{D}  \cap \myset{C}_{b,\xi} $, and  $L_{\mathfrak{g}}\psi_{r-1} \myvar{u}_{\text{nom}} + \alpha(\psi_{r-1}) + L_{\mathfrak{f}}\psi_{r-1} \ge 0 $ is trivially satisfied for $\myvar{x} \in \myset{D}\cap \myset{C}_{b,\xi} $. {Thus}  $\mu$ and $\myvar{u}(\myvar{x})$ are well-defined in $  \myset{D} $.

The solution in \eqref{eq:controllerSol} can be viewed as $\myvar{u}(\myvar{x}) = \omega_1(\myvar{x}) + \omega_2(\omega_3(\myvar{x}))\omega_4(\myvar{x})$
with $\omega_1(\myvar{x}) = \myvar{u}_{\text{nom}}(\myvar{x}), \omega_2(v) = \left\{ \begin{smallmatrix} 
0, & \text{ if } v \ge 0 \\
v, & \text{ if } v < 0 \end{smallmatrix}\right. ,  \omega_3(\myvar{x}) = L_{\mathfrak{g}}\psi_{r-1} \myvar{u}_{\text{nom}} + \alpha(\psi_{r-1}) + L_{\mathfrak{f}}\psi_{r-1}, \omega_4(\myvar{x}) =  \tfrac{- L_{\mathfrak{g}}^\top\psi_{r-1} }{\| L_{\mathfrak{g}}\psi_{r-1} \|^2}$. For $\myvar{x} \in \myset{D} \setminus (\myset{D} \cap \myset{C}_{b,\xi})$, $L_{\mathfrak{g}}\psi_{r-1}(\myvar{x}) \neq \myvar{0}$, we obtain $\omega_1, \omega_2, \omega_3, \omega_4 $ are locally Lipschitz continuous and thus $\myvar{u}(\myvar{x})$ is locally Lipschitz continuous in $\myset{D} \setminus (\myset{D} \cap \myset{C}_{b,\xi})$. Furthermore, for $\myvar{x} \in \myset{D} \cap \myset{C}_{b,\xi}$, we have $\myvar{u}(\myvar{x}) = \myvar{u}_{\text{nom}}(\myvar{x})$ and thus $\myvar{u}(\myvar{x})$ is locally Lipschitz continuous in $\myset{D} \cap \myset{C}_{b,\xi}$.

Now we show that the control input $ \myvar{u}(\myvar{x}) $ is continuous at the boundary between $\myset{D} \cap \myset{C}_{b,\xi}$ and $\myset{D} \setminus (\myset{D} \cap \myset{C}_{b,\xi})$. Assume a Cauchy sequence of points $ \{ \myvar{x}_i \}_{i = 1,2, 3,\cdots} \subset \myset{D} \setminus (\myset{D} \cap \myset{C}_{b,\xi})$ such that $\lim_{i\to \infty } \myvar{x}_i = \myvar{x}_0   $ with $\myvar{x}_0$ at the boundary between $\myset{D} \cap \myset{C}_{b,\xi}$ and $\myset{D} \setminus (\myset{D} \cap \myset{C}_{b,\xi})$. {From} the closed-form solution \eqref{eq:controllerSol} and the facts that $  \myvar{u}_{\text{nom}}(\myvar{x}_i)$ is bounded, $ \lim_{i\to \infty } L_{\mathfrak{g}}\psi_{r-1}(\myvar{x}_i) = \myvar{0},  \lim_{i\to \infty } L_{\mathfrak{f}}\psi_{r-1}(\myvar{x}_i) = 0,$ and $\lim_{i\to \infty }  \alpha(\psi_{r-1}(\myvar{x}_i)) =  \alpha_{r} \circ \alpha_{r-1} \circ \cdots \alpha_{1}(1)  >0 $, we obtain $ \lim_{i\to \infty }   \myvar{u}(\myvar{x}_i) = \myvar{u}(\myvar{x}_0) $. Together with local Lipschitz continuity in $\myset{D} \cap \myset{C}_{b,\xi}$ and $\myset{D} \setminus (\myset{D} \cap \myset{C}_{b,\xi})$, respectively, we conclude {that} the resulting controller from \eqref{eq:controllerQP} is locally Lipschitz continuous. From Theorem \ref{thm:forw_inv_HOCBF}, the resulting controller $\myvar{u}$ guarantees forward invariance of $\myset{C}$.	 
\end{proof}

\subsection{Performance-Critical HOCBF}
In many applications, it would be favorable to know in advance when the nominal controller is implemented without any modifications,i.e., $\myvar{u}(\myvar{x}) = \myvar{u}_{nom}(\myvar{x})$ in some pre-defined set. This is useful, for example, when training a learning-based controller {or performing high-precision motion control during spacecraft rendezvous and docking}. We refer to these instances as ``performance-critical" because, to ensure satisfaction of the task, the designers have to know \textit{a priori} when the nominal control will always be implemented. 

To formally address the performance-critical tasks, we denote the \textit{safety region}\footnote{Note that the safety region may not be the same as the safe set. In the double integrator example, the safe region is the circular region $\myset{C}_{b} = \{ (\myvar{p},\myvar{v}): d^2 - \lVert \myvar{p} \rVert^2 \ge 0  \}  $ that only constrains the state $\myvar{p}$, while the safe set is a subset of $\myset{C}_{b}$ that will be rendered forward invariant.}, inside which the system states should always evolve, and the \textit{performance-critical region}, inside which the nominal control signal should be utilized, as the respective superlevel sets of smooth functions $b, s:\mathbb{R}^n \to \mathbb{R} $. Intuitively, as long as the  performance-critical region lies strictly inside the safety region, with the transformation in \eqref{eq:h_definition},  the nominal control signal is recovered in the performance-critical regions  while safety is always guaranteed.

\begin{thm} \label{thm:nominal_control_recovered}
Consider the control affine system \eqref{eq:nonlinear_dyn}. Let $b,s: \mathbb{R}^n \to \mathbb{R}$ be smooth functions, and let $b$  have least relative degree $r$ in an open set $\myset{D}$ with $\myset{C}_b \subset \myset{D}$. Assume  the conditions in Proposition \ref{prop:lipschitz_continuity_solution}  hold. Assume furthermore that  $\myset{C}_s$ is strictly bounded away from the safety boundary, i.e.,
\begin{equation}
    \myset{C}_s \subseteq \myset{C}_{b,\xi}.
\end{equation}
 Then,  with $h$ given in \eqref{eq:h_definition} and $\psi_{k}$ given in \eqref{eq:a_series_psi},
 \begin{enumerate}
     \item  $h$ is an HOCBF;
     \item  the controller \eqref{eq:controllerQP} renders the set $\myset{C}:= \bigcap_{k =1}^{r}  \myset{C}_{\psi_{k-1}} $ forward invariant for system in \eqref{eq:nonlinear_dyn};
     \item  $\myvar{u}(\myvar{x}) = \myvar{u}_{nom}(\myvar{x})$ for states $\myvar{x}\in  \myset{C}_{s}$.
 \end{enumerate}
\end{thm}

\begin{proof}
      Point 1) and Point 2) follow from Proposition \ref{prop:h_hocbf} and Proposition \ref{prop:lipschitz_continuity_solution}, respectively. For $\myvar{x}\in  \myset{C}_{b,\xi}$, the constraint in the quadratic program \eqref{eq:controllerQP} is trivially satisfied, thus $\myvar{u}(\myvar{x}) = \myvar{u}_{nom}(\myvar{x})$ for states $\myvar{x}\in  \myset{C}_{s}$.
\end{proof}

When  $b$ has exact relative degree $r$ for all states in the safe set, we obtain the following corollary.
\begin{cor}
Consider the control affine system \eqref{eq:nonlinear_dyn}. Let $b,s: \mathbb{R}^n \to \mathbb{R}$ be smooth functions, and let $b$  have uniform relative degree $r$ in an open set $\myset{D}$ with $\myset{C}_b \subset \myset{D}$. Assume that there exists a scalar $\xi > 0$ such that  
\begin{equation}
    \myset{C}_s \subseteq \myset{C}_{b,\xi}.
\end{equation}
Then,   with $h$ given in \eqref{eq:h_definition} and $\psi_{k}$ given in \eqref{eq:a_series_psi},
 \begin{enumerate}
     \item  $h$ is an HOCBF;
     \item  the controller \eqref{eq:controllerQP} renders the set $\myset{C}:= \bigcap_{k =1}^{r}  \myset{C}_{\psi_{k-1}} $ forward invariant for system in \eqref{eq:nonlinear_dyn};
     \item  $\myvar{u}(\myvar{x}) = \myvar{u}_{nom}(\myvar{x})$ for states $\myvar{x}\in \myset{C}_{s}$.
 \end{enumerate}

\end{cor}

\section{An application to rigid-body attitude dynamics}
In this section, we apply the {proposed} high-order control barrier function {methodology} to rigid-body attitude dynamics.  A similar formulation  was proposed in our previous work \cite{tan2020construction}. The main difference is that here we exploit the proposed HOCBF framework to construct a safe, stabilizing control law from a simple nominal stabilizing controller. The method in \cite{tan2020construction} on the other hand uses a more complicated nominal control design. The simulations presented here show that the use of the HOCBF framework allows for modular, safe, stabilizing control design.

The attitude dynamics of a rigid-body with states {consisting of} orientation and angular velocity $(R,\myvar{\omega})$ ((1) in \cite{tan2020construction}) can be written in a control affine form   as
\begin{equation} \label{eq:embedded_dyn}
\dot{\myvar{x}} := \myvarfrak{f}(\myvar{x}) + \myvarfrak{g} \myvar{u},
\end{equation}
where $ \myvar{x} = (r_{11}, r_{12}, \cdots, r_{33}, \omega_1,  \omega_2,  \omega_3 ) \in \mathbb{R}^{12}, \mathfrak{f}(\myvar{x}) = \big(r_{12}\omega_3 - r_{13}\omega_2; r_{13}\omega_1 - r_{11}\omega_3;r_{11}\omega_2 - r_{12}\omega_1; r_{22}\omega_3 - r_{23}\omega_2;r_{23}\omega_1 - r_{21}\omega_3; r_{21}\omega_2 - r_{22}\omega_1; r_{32}\omega_3 - r_{33}\omega_2; r_{33}\omega_1 - r_{31}\omega_3; r_{31}\omega_2 - r_{32}\omega_1; J^{-1}(-[\myvar{\omega}]_{\times} J \myvar{\omega})   \big) \in \mathbb{R}^{12}, \mathfrak{g} =  \begin{pmatrix}
\myvar{0}_{9 \times 3} \\
J^{-1}
\end{pmatrix} $.  We denote  $ C_{TSO(3)} := \{ \myvar{x}\in \mathbb{R}^{12} : \begin{psmallmatrix}
x_{1} & x_{2} & x_{3} \\
x_{4} & x_{5} & x_{6} \\
x_{7} & x_{8} & x_{9} 
\end{psmallmatrix}  \in SO(3) \}  $. In the following, $(R,\myvar{\omega})$ and $\myvar{x}$ are used interchangeably.

Given some sample orientations $R_i\in SO(3), i \in \mathcal{N}$, we define the safe region  $\cup_{i \in \mathcal{N}} S_i$, where $S_i = \{ R \in SO(3): r_i(R) \ge 0  \}, r_i(R) = \epsilon - \| R - R_i \|^2_{F}/2$. Assume that the safe region is connected.
To measure the margin of the attitude trajectory  to the safe region $\cup_{i \in \mathcal{N}} S_i$, we define $b(\myvar{x}) = \sum_{i\in \mathcal{N}} s(r_i(R)/\epsilon) - \delta,$
where $ \delta >0$ is a constant,  and a smooth transition function $s(v) = \left\{ \begin{smallmatrix}
0 & v\in ( -\infty,0), \\
 \tfrac{\rho(v)}{\rho(v) + \rho(1 - v)} & v \in [0,1), \\
1 & v \in [1,\infty)
\end{smallmatrix}\right.
$ with $ \rho(v) := (1/v)e^{-1/v} $.  The associated constrained set is $ \myset{C}_b(\myvar{x}):=\{ \myvar{x} \in \myset{C}_{TSO(3)}: b(\myvar{x})\ge 0 \} $. To ensure that the trajectory evolves within $\cup_{i\in \mathcal{N}} S_i$, we conservatively require $b(\myvar{x}(t)) \ge 0$ for $t\ge 0$.

\begin{figure}[ht]
	\centering
	\begin{subfigure}[t]{0.3\linewidth}
		\includegraphics[width=\linewidth]{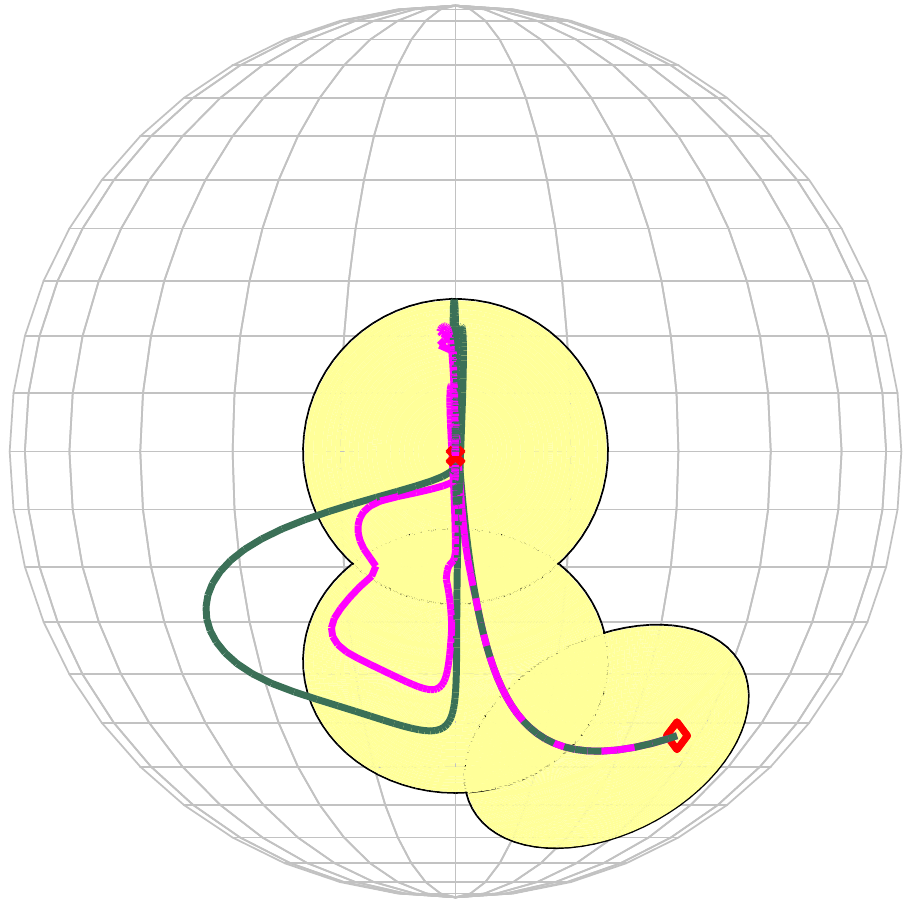}
		\caption{   $x$-axis. }   
	\end{subfigure}
	\begin{subfigure}[t]{0.3\linewidth}
		\centering\includegraphics[width=\linewidth]{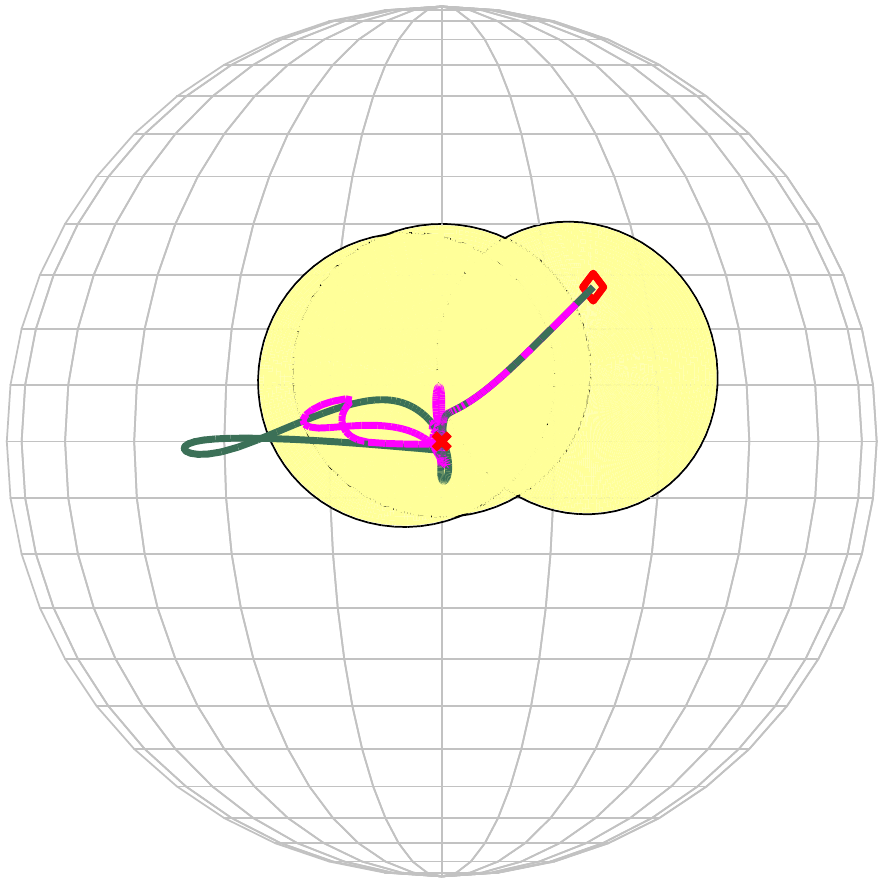}
		\caption{ $y$-axis.}
	\end{subfigure}
	\begin{subfigure}[t]{0.3\linewidth}
		\centering\includegraphics[width=\linewidth]{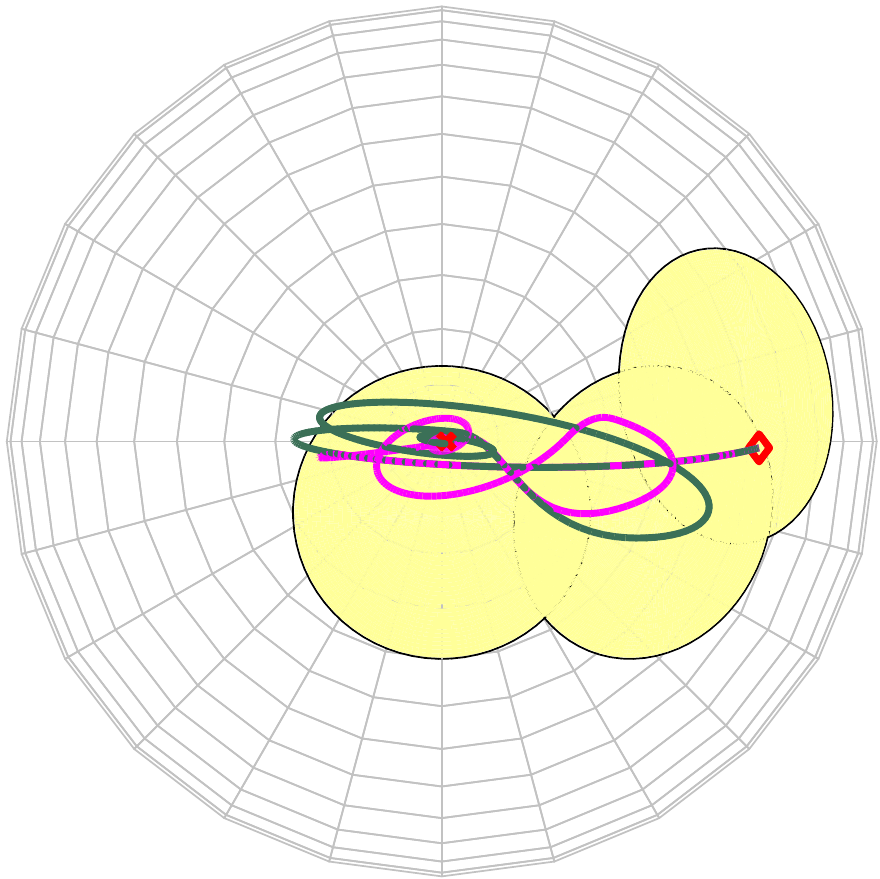}
		\caption{  $z$-axis.}
	\end{subfigure}
	\setlength{\belowcaptionskip}{-18pt}
	\caption{  Comparison of the attitude trajectories in body-fixed $xyz$ axes with additive control signals. The square point and  the cross point represent the starting attitude $R_0$ and the target attitude $R_f$, respectively, and the yellow region is the safe region.  The purple and green lines represent the results {wherein} the barrier function is in use or not in use, respectively.  }
	\label{fig:simulated_trajectory_human_input}
	
    \end{figure}

\begin{figure}[!h]
    \centering
    \begin{subfigure}[t]{\linewidth}
    \includegraphics[width=\linewidth]{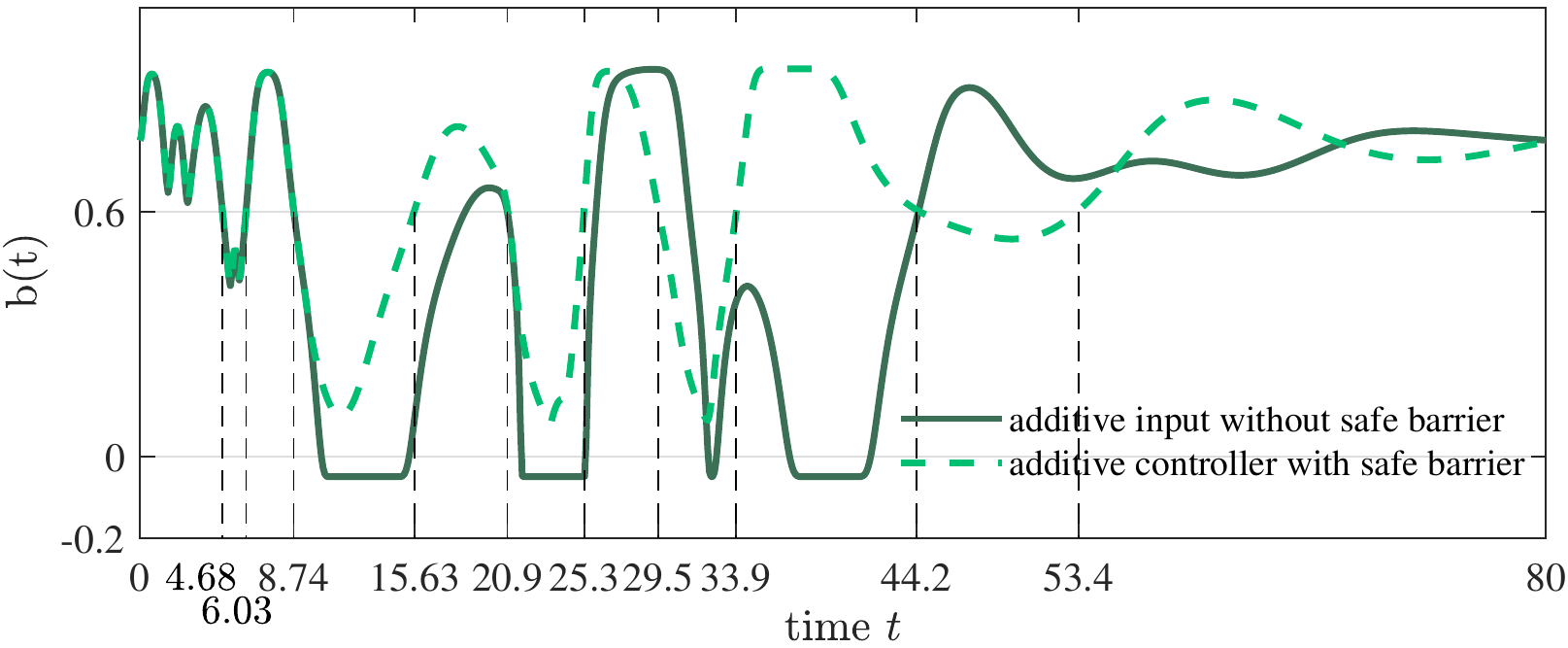}
    \caption{Time histories of $b(t)$ when additive control signals exist. For all $t$ with $b(t)>\xi = 0.6$, the system state is in the performance-critical region where the nominal control signal is used.}
    \end{subfigure}
    \begin{subfigure}[t]{\linewidth}
    \includegraphics[width=\linewidth,height =\linewidth/2 ]{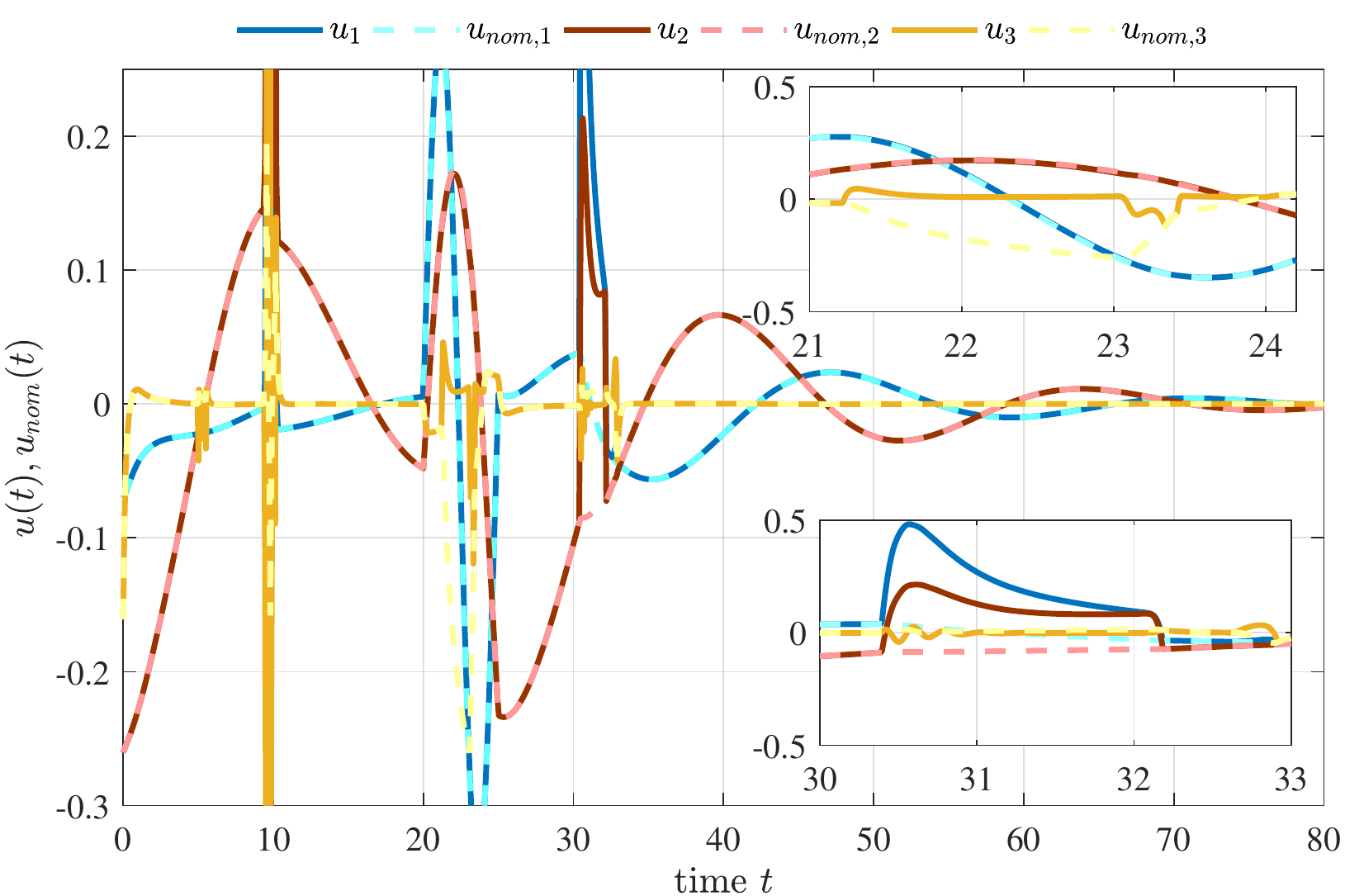}
    \caption{The time history of the nominal and actual control inputs when the barrier function is in use. The discrepancies between $\myvar{u}(t)$ and $\myvar{u}_{nom}(t)$ occur in the time interval $t\in[4.8, 5.6]\cup[9.4,10.4]\cup[21,24]\cup[30,33]$ and $b(t) < \xi = 0.6$ for all $t$ in this interval.}
    \end{subfigure}
    \caption{Attitude stabilization with additive control signals.}
    \label{fig:additive_control_scn}
\end{figure}

Following the analysis in \cite{tan2020construction}, we know that $b(\myvar{x})$ is of least relative degree $r = 2$. Moreover, the singular points, at which the exact relative degree is greater than $2$, lie on the geodesics between the sampling points $R_i, i \in \mathcal{N}$ \cite[Proposition~3]{tan2020construction}, {and} thus are bounded away from the boundary of the safe region. This fact satisfies the assumption in Proposition \ref{prop:h_hocbf}. Applying the results in this paper, we obtain: 1) $ h(\myvar{x}) =  \chi\left(\tfrac{b(\myvar{x})}{\xi} \right)$
is an HOCBF; 2) with $h(\myvar{x})$ given, the set $\myset{C} = \myset{C}_{\psi_0} \cap \myset{C}_{\psi_1}$ is forward invariant; 3)  the nominal control signal will be implemented  in any subset of $\myset{C}_{b,\xi}$ (performance-critical set).

 We consider an attitude stabilization scenario from $(R_0, \myvar{0})$ to $(R_f, \myvar{0}), J = \begin{psmallmatrix}
5.5   & 0.06 & -0.03\\
0.06 &  5.5   & 0.01\\
-0.03 & 0.01 & 0.1
\end{psmallmatrix}  \textup{kg}\cdot\textup{m}^2. $ We set $ R_f = I $, the sampling orientations $ R_3= \exp(10^{\circ}/{180^\circ} \times {\pi}[\myvar{e}_1]_{\times})$, $ R_2= \exp(30^{\circ}/{180^\circ} \times {\pi}[\myvar{e}_2]_{\times})R_3$, $ R_1 = \exp(30^{\circ}/{180^\circ} \times {\pi}[0, 0.447,0.894]_{\times})R_2 $, the initial attitude $ R_0  =  \exp(10^{\circ}/{180^\circ} \times {\pi}[\myvar{e}_1]_{\times})R_1 $, and $\epsilon = 0.1206 $, which corresponds to cell radius  $0.3491 $ rad $(20^\circ)  $.  We use the saturated stabilizing controller from \cite{lee2012robust} as the nominal controller:
\begin{equation} \label{eq:u_nom}
       \myvar{u}_{nom}(R,\myvar{\omega}) =  -k_1(R- R^\top )^{\vee} -k_2 \tanh(\myvar{\omega}),
\end{equation}
where  $\tanh(\cdot)$ is the element-wise hyperbolic tangent function.  The controller parameters are set as $ k_1 = k_2 = 0.2$. The parameters in the control barrier function are chosen as $ \delta = 0.05,\xi = 0.6,\alpha_1(v) = \alpha_2(v) = v, \chi(v) =\left\{ \begin{smallmatrix} 
(v-1)^3+1, &   \text{if }v \le 1;\\
 1,  & \text{if }v > 1. 
\end{smallmatrix}\right. $

We simulate an attitude stabilization scenario where the control signal in \eqref{eq:u_nom} is augmented with an additive signal $ \myvar{u}_{add} = 0.3* \big(\sin(2\pi \tfrac{t-20}{5}), \sin(\pi \tfrac{t-20}{5}),-\sin(\pi \tfrac{t-20}{5}) \big) $ for the time interval $ t\in [20,25] $ and view their sum as the nominal control signal {in} the quadratic program \eqref{eq:controllerQP}. This control signal simulates, for example, a human input to the system that leads to a deviation from the previous trajectory and may drive the states out of the safe region. The trajectories are shown in Fig. \ref{fig:simulated_trajectory_human_input}. When the barrier function is in use, the resulting trajectory evolves within the  safe region. Moreover, from Fig. \ref{fig:additive_control_scn}, we see that the actual control signal coincides with the nominal control signal whenever $ b(t)\geq \xi = 0.6$, which validates the performance-critical property.

Compared to the simulation results in \cite{tan2020construction}, we note that similar results are obtained here with a simple nominal stabilizing control law. This shows the effectiveness and modularity of the proposed HOCBF framework.

\section{Conclusion}

In this paper, we formulate high-order (zeroing) barrier functions {and their controlled equivalent} for nonlinear dynamical {systems}. This formulation generalizes the concept of zeroing barrier functions and similar concepts in the literature. {Our results do not require}  forward completeness of the system to show forward invariance {of the set}. More importantly, we show for the first time that the {intersection of superlevel sets} associated with the high-order barrier function, is asymptotically stable. Thanks to this property, our method generalizes the robustness results of the standard zeroing barrier function formulation. We also provide a remedy to handle the singular states that arise when implementing the minimally-invasive control law, while ensuring safety of the overall system. Finally, we derive a performance-critical property so that one can define the performance-critical regions \textit{a priori}. The proposed formulation is implemented on the non-trivial case study of rigid-body attitude dynamics.

\bibliographystyle{IEEEtran}
\bibliography{IEEEabrv,main}

\end{document}